\documentclass[draft, letterpaper, romanappendices, onecolumn]{IEEEtran}
\usepackage{graphics}
\usepackage{cite}
\usepackage[pdftex]{graphicx}
\usepackage{epstopdf}
\usepackage{epsfig}
\usepackage{latexsym}
\usepackage{amsfonts}
\usepackage{calc}
\usepackage{url}
\usepackage{enumerate}
\usepackage{color}
\usepackage[tbtags]{amsmath}
\usepackage{amssymb}
\usepackage{upref}
\usepackage{dsfont}
\usepackage{multirow}
\usepackage{booktabs}
\usepackage{bigstrut}
\usepackage{rotating}
\usepackage{nth}
\usepackage{breqn}

\newtheorem{theorem}{Theorem}

\newtheorem{lemma}{Lemma}

\newtheorem{proof}[theorem]{Proof}

\begin{document}
\title{Coding Delay Analysis of Chunked Codes \\ over Line Networks${}^{\dagger}$\thanks{$^{\dagger}$This paper is an extended version of an accompanying work submitted to NetCod 2012.}}

\author{\IEEEauthorblockN{Anoosheh~Heidarzadeh and Amir H. Banihashemi\\}
\IEEEauthorblockA{\small{Department of Systems and Computer Engineering, Carleton University, Ottawa, ON, Canada}\\
%\small{Email: {\{anoosheh,ahashemi\}@sce.carleton.ca}}}
}\vspace{-0.5 cm}}

\maketitle
\thispagestyle{empty}

\begin{abstract}In this paper, we analyze the coding delay and the average coding delay of Chunked network Codes (CC) over line networks with Bernoulli losses and deterministic regular or Poisson transmissions. Chunked codes are an attractive alternative to random linear network codes due to their lower complexity. Our results, which include upper bounds on the delay and the average delay, are the first of their kind for CC over networks with such probabilistic traffics. These results demonstrate that a stand-alone CC or a precoded CC provides a better tradeoff between the computational complexity and the convergence speed to the network capacity over the probabilistic traffics compared to arbitrary deterministic traffics. The performance of CC over the latter traffics has already been studied in the literature.\end{abstract}

\vspace{-.35 cm}
\section{Introduction}
Chunked codes (CC), originally proposed in~\cite{MHL:2006}, generalize random linear network codes (dense codes), and operate by dividing the message of the source into non-overlapping or overlapping sub-messages of equal size, called \emph{chunks}~\cite{MHL:2006,SZK:2009,HBJ:2011}. Each node at each transmission time randomly chooses a chunk, and transmits it by using a dense code. In fact, a dense code is a CC with only one chunk of the size equal to the message size. Thus, CC require less complex coding operations due to applying coding on chunks smaller than the original message. This however comes at the cost of lower speed of convergence to the capacity compared to dense codes.

The speed of convergence of CC to the capacity of line networks with arbitrary deterministic traffics was studied in~\cite{MHL:2006,HBJ:2011}. In particular, it has been shown that (i) a CC achieves the capacity, so long as the size of the chunks is lower bounded by a function super-logarithmic in the message size and super-log-cubic in the network length, and (ii) a CC, preceded by a capacity-achieving erasure code, approaches the capacity with an arbitrarily small but non-zero constant gap, so long as the size of the chunks is lower bounded by a function constant in the message size and log-cubic in the network length. There is however no result on the speed of convergence of CC to the capacity over the networks with probabilistic traffics.

The speed of convergence of dense codes to the capacity of some probabilistic traffics was studied in~\cite{PFS:2005,DDHE:2009}. Very recently, in~\cite{HB1S:2012}, we studied the coding delay and the average coding delay of a dense code over the traffics with deterministic regular or Poisson transmissions and Bernoulli losses.\footnote{The \emph{coding delay} of a code over a network with a given traffic (schedule of transmissions and losses) is the minimum time that the code takes to transmit all the message vectors from the source to the sink. The coding delay is a random variable due to the randomness in both the code and the traffic. The \emph{average coding delay} of a code with respect to a class of traffics is the coding delay of the code averaged out over all the traffics (but not the codes), and hence is a random variable due to the randomness in the code.} The results were in some cases more general, and in some other cases tighter, than the existing bounds in~\cite{PFS:2005,DDHE:2009}.

%

% There is however no result on the coding delay or the average coding delay of CC over such network schedules in the existing literature.

% , which are two parameters often used to measure the speed of convergence of a coding scheme to the capacity of a network,

In this paper, we generalize our analysis in~\cite{HB1S:2012}, and for the first time, study the coding delay and the average coding delay of CC for different ranges of the chunk sizes.\footnote{In this paper, we focus on CC with non-overlapping chunks. The analysis of CC with overlapping chunks is the focus of an ongoing research project.}

%The networks, we consider, experience traffics with deterministic regular or Poisson transmission schedules and Bernoulli losses.

The main contributions of this work are:
\begin{itemize}
\item We derive upper bounds on the coding delay and the average coding delay of a CC alone, or a CC with precoding, over the traffics with deterministic regular transmissions or Poisson transmissions and Bernoulli losses with arbitrary parameters or unequal parameters.
\item We show that: (i) a CC achieves the capacity, so long as the size of the chunks is bounded from below by a function super-logarithmic in the message size and super-log-linear in the network length, and (ii) the combination of a CC and a capacity-achieving erasure code approaches the capacity with an arbitrarily small non-zero constant gap, so long as the size of the chunks is bounded from below by a function constant in the message size and log-linear in the network length. The lower bounds in both cases are smaller than those over the networks with arbitrary deterministic traffics. Thus both coding schemes are less computationally complex (require smaller chunks), for the same speed of convergence, over such probabilistic traffics, compared to arbitrary deterministic traffics.
\item In a capacity-achieving scenario, for such probabilistic traffics, we show that: (i) the upper bound on the overhead\footnote{The (\emph{average}) \emph{overhead} is the difference between the (average) coding delay and the ratio of the message size to the capacity.} grows sub-log-linearly with the message size and the network length, and decays sub-linearly with the size of the chunks, and (ii) the upper bound on the average overhead grows sub-log-linearly (or poly-log-linearly) with the message size, and sub-log-linearly (or log-linearly) with the network length, and decays sub-linearly (or linearly) with the size of the chunks, in the case with arbitrary (or unequal) parameters. For arbitrary deterministic traffics, the upper bound on the overhead was shown in~\cite{HBJ:2011} to be similar to (i), but with a larger (super-linear) growth rate with the network length.
\end{itemize}
%\footnote{The overhead is the difference between the coding delay and the ratio of the message size to the capacity.}
%\footnote{A function $f(n)$ is called ``log-linear'' in $n$ if $f(n)=O(n\log n)$.}
%\footnote{The coding cost is the ratio of the number of coding (packet) operations to the message size.}

\vspace{-.25 cm}
\section{Network Model and Problem Setup}
\subsection{Transmission and Loss Model}
We consider a unicast problem (one-source one-sink) over a line network with $L$ links connecting $L+1$ nodes $\{v_i\}_{0\leq i\leq L}$ in tandem. The source node $v_0$ has a message of $k$ vectors (called \emph{message vectors}) from a vector space $\mathcal{F}$ over $\mathbb{F}_2$, and the sink node $v_L$ requires all the message vectors.\footnote{The analysis in this paper is generalizable to finite fields of larger size.}

Each (non-sink) node at each transmission time transmits a (coded) packet, which is a vector in $\mathcal{F}$. The packet transmissions are assumed to occur in discrete-time, and the transmission times over different links are assumed to follow independent stochastic processes. The transmission times over the $i\textsuperscript{th}$ link are specified by (i) a deterministic process where there is a packet transmission at each time instant, or (ii) a Poisson process with parameter $\lambda_i: 0<\lambda_i\leq 1$, where $\lambda_i$ is the average number of transmissions per time unit over the $i\textsuperscript{th}$ link. The transmission schedules resulting from (i) and (ii) are referred to as \emph{deterministic regular} and \emph{Poisson}, respectively.

Each transmitted packet either succeeds or fails to be received (\emph{successful} vs. \emph{lost}). The successful packets are assumed to arrive with zero delay, and the lost packets will never arrive. The packets are assumed to be successful independently over different links. The successful packets over the $i\textsuperscript{th}$ link are specified by a Bernoulli process with (success) parameter $p_i: 0<p_i\leq 1$, where $p_i$ is the average number of successes per transmission over the $i\textsuperscript{th}$ link. The loss model defined as above is referred to as \emph{Bernoulli}. The special case of Bernoulli loss with all $p_i$'s equal to $1$ is analogous to the lossless case.

\vspace{-.35 cm}
\subsection{Problem Setup}
The goal in this paper is to derive upper bounds on the coding delay and the average coding delay of chunked codes over line networks with deterministic regular or Poisson transmissions and Bernoulli losses.

In a chunked coding scheme, the set of $k$ message vectors at the source node is divided into $q$ disjoint subsets, called \emph{chunks}, each of size $\alpha=k/q$. The source node, at each transmission time, chooses a chunk independently at random, and transmits a packet by randomly linearly combining the message vectors belonging to the underlying chunk. Each non-source non-sink node, at the time of each transmission, chooses a chunk independently at random, and transmits a packet by randomly linearly combining its previously received packets pertaining to the underlying chunk. The global encoding vector\footnote{The \emph{global encoding vector} of a packet is the vector of the coefficients representing the mapping between the message vectors and the packet.} of each packet is assumed to be transmitted along with the packet. The sink node can decode a chunk, so long as it receives an innovative\footnote{A collection of packets is \emph{innovative} if the global encoding vectors of the packets belonging to the collection are linearly independent.} collection of packets pertaining to the underlying chunk of a size equal to the size of the chunk.

\vspace{-.25 cm}
\section{Deterministic Regular Transmissions and Bernoulli Losses}\label{sec:BernoulliLossRegularTrafficCC}
We first review the analysis of dense codes, which are a special case of CC with one chunk, in two cases of arbitrary or unequal (success) parameters, presented in~\cite{HB1S:2012}.\footnote{The details of the proofs in the case of arbitrary parameters were given in~\cite{HB1S:2012} and hence omitted. However, neither the details, nor the sketches of the proofs in the case of unequal parameters were given in~\cite{HB1S:2012}. We present the sketches of the proofs in this paper for the purpose of completeness.} Next, we generalize the analysis to CC with more than one chunk.

\vspace{-.35cm}
\subsection{Dense Codes}\label{subsec:DC}
The goal of the analysis is to lower bound (i) the size of a maximal dense collection of packets at the sink node until a certain time,\footnote{A collection of packets is \emph{dense} if the local encoding vectors of the packets are linearly independent, where the local encoding vector of a packet is the vector of the coefficients of the linear combination pertaining to the packet.} and then, (ii) the probability that a sufficient number of packets in the underlying collection are innovative.

Let $Q_{i+1}$ and $Q_i$ be the decoding matrices\footnote{The global encoding vectors of the packets at a node form the rows of the \emph{decoding matrix} at the node.} at the $(i+1)\textsuperscript{th}$ and $i\textsuperscript{th}$ nodes, respectively, and $T_i$ be a matrix over $\mathbb{F}_2$ such that $Q_{i+1}=T_i Q_i$. The entries of $Q_{i+1}$ and $Q_i$ are in $\mathbb{F}_2$. Each row of $T_i$ is the local encoding vectors of a successful packet sent by the $i\textsuperscript{th}$ node. Let $Q'_i$ be $Q_i$ restricted to its rows corresponding to the global encoding vectors of the dense packets at the $i\textsuperscript{th}$ node. Let $T'_i$, the \emph{transfer matrix} at the $i\textsuperscript{th}$ node, be a matrix over $\mathbb{F}_2$ such that $Q_{i+1}=T'_i Q'_i$. Each row of $T'_i$ indicates the labels of the dense packets at the $i\textsuperscript{th}$ node which contribute to a successful packet sent by the $i\textsuperscript{th}$ node.

For every matrix $Q$ over $\mathbb{F}_2$, the \emph{density} of $Q$, denoted by $\mathcal{D}(Q)$, is the size of a maximal dense collection of rows in $Q$, where a collection of rows is \emph{dense} if the rows have all independent and uniformly distributed Bernoulli entries. Further, $Q$ is called a \emph{dense matrix} if all its rows form a dense collection. For every matrix $T$ over $\mathbb{F}_2$, the \emph{rank} of $T$, denoted by $\text{rank}(T)$, is the size of a maximal collection of linearly independent rows in $T$.

\begin{lemma}\label{lem:DensityTM} Let $Q$ be a dense matrix over $\mathbb{F}_2$, and $T$ be a matrix over $\mathbb{F}_2$, where the number of rows in $Q$ and the number of columns in $T$ are equal. If $\text{rank}(T)\geq \gamma$, then $\mathcal{D}(TQ)\geq \gamma$.\end{lemma}

Since $Q_{i+1}=T'_i Q'_i$, and $Q'_i$ is dense, $\mathcal{D}(Q_{i+1})$ is lower bounded so long as $\text{rank}(T'_i)$ is lower bounded. As shown in~\cite{HB1S:2012}, the matrix $T'_i$ includes a sub-matrix with the structure of a random block lower-triangular matrix, and the rank of a matrix with such a structure is lower bounded as follows.

Let $w$, $r$ and $\{r_j\}_{1\leq j\leq w}$ be arbitrary non-negative integers, and let $r_{\text{max}}=\max_j r_j$ and $r_{\text{min}}=\min_j r_j$. Let $T_{i,j}$ be an $r\times r_j$ dense matrix over $\mathbb{F}_2$, if $1\leq j\leq i\leq w$; or an arbitrary $r\times r_j$ matrix over $\mathbb{F}_2$, otherwise. Let $T=[T_{i,j}]_{1\leq i,j\leq w}$. The matrix $T$ is called \emph{random block lower-triangular} (RBLT).

\begin{lemma}\label{lem:VerticalT} Let $T$ be an RBLT matrix with parameters $w$, $r$ and $\{r_j: 0\leq r_j\leq r\}_{1\leq j\leq w}$. Let $u=\left\lceil{(n-\gamma)}/{r_{\text{min}}}\right\rceil$, and $n= \sum_{1\leq j\leq w}r_j$. For every integer $0\leq \gamma\leq n-1$, $\Pr\{r(T)<n-\gamma\}\leq u \left(1-2^{-r_{\text{max}}}\right) 2^{-\gamma+n-wr+(r-r_{\text{min}})(u-1)}$.\end{lemma}

\begin{lemma}\label{lem:HorizontalT} Let $T$ be an RBLT matrix with parameters $w$, $r$ and $\{r_j: 0\leq r\leq r_j\}_{1\leq j\leq w}$. Let $u=\left\lceil{(n-\gamma)}/{r}\right\rceil$, and $n = w r$. For every integer $0\leq \gamma\leq n-1$, $\Pr\{r(T)<n-\gamma\}\leq u \left(1-2^{-r}\right) 2^{-\gamma+n-wr_{\text{min}}+(r_{\text{min}}-r)(u-1)}$.\end{lemma}

The application of the lemmas is subject to useful and tight choices of $w$, $r$, and $r_j$'s. Such parameters depend on the traffic over the $i\textsuperscript{th}$ and $(i+1)\textsuperscript{th}$ links, and hence not straightforward to optimize. However, by using a probabilistic technique, tight bounds on such parameters can be derived.

Let $(0,N_T]$ be the period of time over which the transmissions occur. Let $(0,N_T]$ be divided into $w$ disjoint partitions of length $N_T/w$. For every $1\leq j\leq w$, and $1\leq i\leq L$, let $I_{ij}$ be the $j\textsuperscript{th}$ partition pertaining to the $i\textsuperscript{th}$ link. For every $i,j$: $i\leq j\leq w-L+i$, $I_{ij}$ is called \emph{active}. Let $w_T\doteq L(w-L+1)$ be the total number of active partitions.

Let $\varphi_{ij}$ be the number of successful packets in $I_{ij}$. By the assumption, $\varphi_{ij}$ is a binomial random variable with the expected value $\varphi_i=p_i N_T/w$. Let $p\doteq\min_{1\leq i\leq L}p_i$, and $\varphi\doteq pN_T/w$. For any real number $x$, let $\dot{x}$ denote $x/2$. By applying the Chernoff bound, one can show that $\varphi_{ij}$ is not larger than or equal to $r\doteq (1-\gamma^{*})\varphi$ with probability (w.p.) bounded above by (b.a.b.) $\dot{\epsilon}/w_T$, so long as $0<\gamma^{*}<1$, where $\gamma^{*}\sim \sqrt{({1}/{\dot{\varphi}})\ln({w_T}/{\dot{\epsilon}})}$. For all $i,j$, suppose that $\varphi_{ij}$ is larger than or equal to $r$.

% \footnote{For any real number $x$, let $\dot{x}$ denote $x/2$.}

% Such a choice of $\gamma^{*}$ goes to $0$, as $N_T$ goes to infinity, so long as \[w\log\frac{w_T}{\epsilon}=o(p N_T).\]

Let $\mathcal{D}(Q_i^j)$ be the number of dense packets in the first $j$ active partitions over the $i\textsuperscript{th}$ link.

The packets over the first link are all dense. Thus, for all $j$, $\mathcal{D}(Q_1^j)\geq rj$. For any other values of $i,j$, by applying Lemma~\ref{lem:VerticalT}, it can be shown that the inequality $\mathcal{D}(Q_i^j)\geq rj-j(1+o(1))\log(w_T/\epsilon)$ fails w.p. b.a.b. $ij\dot{\epsilon}/w_T$, so long as \begin{equation}\label{eq:Temp7}w\log\frac{w_T}{\epsilon}=o(p N_T).\end{equation}

This result shows that the number of dense packets at the sink node, $\mathcal{D}(Q_L)$, fails to be larger than \begin{eqnarray}\label{eq:Temp8}
 \lefteqn{p N_T - O(p N_T L/w) -  } \nonumber\\
   && O(\sqrt{p N_T w\log(w L/\epsilon)})-O(w\log(w L/\epsilon)),
\end{eqnarray} w.p. b.a.b. $\epsilon$. By condition~\eqref{eq:Temp7}, it follows that each $O(.)$ term in~\eqref{eq:Temp8} is $o(p N_T)$ which ensures that the code achieves the capacity. We specify $w$ by $\sqrt[3]{{p N_T L^2}/{\log(p N_T L/\epsilon)}}$ in order to maximize~\eqref{eq:Temp8} subject to condition~\eqref{eq:Temp7}.

Let $n_T$ be equal to~\eqref{eq:Temp8}. Thus, $Q_L$ fails to include an $n_T\times k$ dense sub-matrix w.p. b.a.b. $\epsilon$.

\begin{lemma}\label{lem:DenseRankProb}Let $Q$ be an $n\times k$ ($k\leq n$) dense matrix over $\mathbb{F}_2$. Then, $\Pr\{\text{rank}(Q)<k\}\leq 2^{-(n-k)}$.\end{lemma}

By Lemma~\ref{lem:DenseRankProb}, $\Pr\{\text{rank}(Q_L)<k\}$ is b.a.b. $\epsilon$, so long as $k\leq n_T-\log(1/\epsilon)$. By replacing $\epsilon$ with $\dot{\epsilon}$, it follows that the sink node can recover all the message vectors w.p. b.a.b. $\epsilon$, so long as $k\leq n_T - \log(1/\epsilon)-1$. Let $k_{\text{max}}$ be the largest integer $k$ satisfying this inequality. Thus, $k_{\text{max}}\sim p N_T$, and by replacing $N_T$ with $k/p$, the following result is immediate.

\begin{theorem}\label{thm:DenseCodesRegularBernoulliActualNon-IdenticalGeneral} The coding delay of a dense code over a line network of $L$ links with deterministic regular traffics and Bernoulli losses with parameters $\{p_i\}$ is larger than \[\frac{1}{p}\left(k+(1+o(1))\left(\frac{kL}{w}+\sqrt{k\left(w\log\frac{wL}{\epsilon}\right)}+w\log\frac{wL}{\epsilon}\right)\right)\] w.p.~b.a.b.~$\epsilon$,~where~$w\sim \left(k L^2/\log(k L/\epsilon)\right)^{\frac{1}{3}}$,~$p\doteq\min_{1\leq i\leq L}{p_i}$.\end{theorem}

In the case of the average coding delay, the analysis proceeds by replacing $r$ with $\varphi$ in the preceding results, and re-specifying $w$ by $\sqrt{{p N_T L}/{\log(p N_T L/\epsilon)}}$ in order to maximize \begin{equation}\label{eq:Temp5} pN_T - O(p N_T L/w) - O(w\log(w L/\epsilon)),\end{equation} instead of~\eqref{eq:Temp8}, subject to condition~\eqref{eq:Temp7}.

\begin{theorem}\label{thm:DenseCodesRegularBernoulliAverageNon-IdenticalGeneral} The average coding delay of a dense code over a network similar to Theorem~\ref{thm:DenseCodesRegularBernoulliActualNon-IdenticalGeneral} is larger than \[\frac{1}{p}\left(k +(1+o(1))\left(\frac{k L}{w}+w\log\frac{wL}{\epsilon}\right)\right)\] w.p. b.a.b. $\epsilon$, where $w\sim \left({kL/\log(kL/\epsilon)}\right)^{\frac{1}{2}}$.\end{theorem}

In order to derive tighter bounds the actual values of the success parameters $\{p_i\}$ need to be taken into consideration. In particular, the coding delay and the average coding delay of dense codes for a special case, where no two links have equal success parameters, are upper bounded as follows.

Let us assume $p_1>p_2>\cdots> p_L$, without loss of generality. Let $p\doteq \min_{1\leq i\leq L} p_i$, $\gamma_e\doteq \min_{1<i\leq L}\gamma_{e_i}$, and $\gamma_{e_i}\doteq |p_i-p_{i-1}|$. Let $r_i\doteq (1-\gamma^{*}_i)\varphi_i$, where $\varphi_i=p_iN_T/w$ and $\gamma^{*}_i\sim\sqrt{(1/\dot{\varphi_i})\log(w_T/\dot{\epsilon})}$. Let $\varphi_{ij}$ be defined as before. For all $i,j$, suppose that $\varphi_{ij}$ is larger than or equal to $r_i$.

Similarly as before, for all $j$, $\mathcal{D}(Q_1^j)\geq r_1 j$. For any other values of $i,j$, by applying Lemma~\ref{lem:HorizontalT}, it can be shown that the inequality $\mathcal{D}(Q_i^j)\geq r_i j$ fails w.p. b.a.b. $ij\dot{\epsilon}/w_T$, so long as \begin{equation}\label{eq:UnequalParameters} {w}\log\frac{w_T}{\epsilon}=o(\gamma_e p N_T).\end{equation}

Let $p$, $\varphi$, $\gamma^{*}$ and $r$ denote $p_L$, $\varphi_L$, $\gamma^{*}_L$ and $r_L$, respectively. Thus, the inequality $\mathcal{D}(Q_L)\geq (1-\gamma^{*})\varphi w_T/L$ fails w.p. b.a.b. ${\epsilon}$. By replacing $\varphi$ with $pN_T/w$, the right-hand side of the last inequality can be written as: \begin{equation}\label{eq:Temp4} p N_T - O(p N_T L/w) -O(\sqrt{p N_T w \log(w L/\epsilon)}).\end{equation} The rest of the analysis is similar to that of Theorem~\ref{thm:DenseCodesRegularBernoulliActualNon-IdenticalGeneral}, except that~\eqref{eq:Temp4} excludes the last term in~\eqref{eq:Temp8}, and the choice of $w$ needs to satisfy condition~\eqref{eq:UnequalParameters}, instead of condition~\eqref{eq:Temp7}.

\begin{theorem}\label{thm:DenseCodesRegularBernoulliActualNon-Identical} Consider a sequence of unequal parameters $\{p_i\}_{1\leq i\leq L}$. The coding delay of a dense code over a line network of $L$ links with deterministic regular traffics and Bernoulli losses with parameters $\{p_i\}$ is larger than \[\frac{1}{p}\left(k+(1+o(1))\left(\frac{kL}{w}+\sqrt{k\left(w\log\frac{wL}{\epsilon}\right)}\right)\right)\] w.p. b.a.b. $\epsilon$, where $w\sim\gamma_e \left(k L^2/\log(k L/\epsilon)\right)^{\frac{1}{3}}$, $p \doteq \min_{1\leq i\leq L}p_i$, $\gamma_{e}\doteq \min_{1< i\leq L} \gamma_{e_i}$, and $\gamma_{e_i} \doteq |p_i-p_{i-1}|$.\end{theorem}

In the case of the average coding delay, the analysis follows the same line as that of Theorem~\ref{thm:DenseCodesRegularBernoulliAverageNon-IdenticalGeneral}, except that the choice of $w$ needs to maximize \begin{equation}\label{eq:Temp9} pN_T - O(p N_T L/w)\end{equation} subject to condition~\eqref{eq:UnequalParameters}, instead of~\eqref{eq:Temp5} subject to condition~\eqref{eq:Temp7}.

\begin{theorem}\label{thm:DenseCodesRegularBernoulliAverageNon-Identical} The average coding delay of a dense code over a network similar to Theorem~\ref{thm:DenseCodesRegularBernoulliActualNon-Identical} is larger than \[\frac{1}{p}\left(k +(1+o(1))\left(\frac{k L}{w}\right)\right)\] w.p. b.a.b. $\epsilon$, where $w\sim \gamma_e k/(f(k)\log(k L/\epsilon))$, and $f(k)$ goes to infinity, as $k$ goes to infinity, such that $f(k)=o(\gamma_e k/\log(kL/\epsilon))$.\end{theorem}

\vspace{-.35 cm}
\subsection{CC: Capacity-Achieving}\label{subsec:CCCACH}
In a CC, at each transmission time, a chunk is chosen w.p. $1/q$, and a packet transmission over the $i\textsuperscript{th}$ link is successful w.p. $p_i$. Thus the probability that a given packet transmission over the $i\textsuperscript{th}$ link is successful and pertains to a given chunk is $p_i/q$. Thus by replacing $p_i$ with $p_i/q$ in the analysis of dense codes in Section~\ref{subsec:DC}, the coding delay and the average coding delay of CC in a capacity-achieving scenario will be upper bounded as follows.

The results of dense codes are indeed a special case of those of CC with one chunk of size $k$. It is, however, worth noting that, due to the change in the parameters, the number of partitions $w$ needs to satisfy a new condition: $wq\log\frac{w_T q}{\epsilon}=o(p N_T)$ or $wq\log\frac{w_T q}{\epsilon}=o(\gamma_e p N_T)$, instead of condition~\eqref{eq:Temp7} or~\eqref{eq:UnequalParameters}, in the proofs of Theorems~\ref{thm:CapAchCodDelGeneral} and~\ref{thm:CapAchAveCodDelGeneral}, or those of Theorems~\ref{thm:CapAchCodDelSpecial} and~\ref{thm:CapAchAveCodDelSpecial}, respectively. Further by replacing $w$ with its optimal choice in the new version of~\eqref{eq:Temp8},~\eqref{eq:Temp5},~\eqref{eq:Temp4} and~\eqref{eq:Temp9}, each $O(.)$ term needs to be $o(pN_T/q)$ in order to ensure that CC are capacity-achieving in the underlying case. Such a condition lower bounds the size of the chunks $\alpha$ by a function super-logarithmic in the message size $k$.

\begin{theorem}\label{thm:CapAchCodDelGeneral} The coding delay of a CC with $q$ chunks over a line network of $L$ links with deterministic regular traffics and Bernoulli losses with parameters $\{p_i\}$ is larger than \begin{dmath*}\frac{1}{p}\left(k+(1+o(1))\left(\frac{k L}{w}+\sqrt{k\left(wq\log\frac{w q L}{\epsilon}\right)}+wq\log\frac{w q L}{\epsilon}\right)\right)\end{dmath*} w.p. b.a.b. $\epsilon$, so long as $q=o({k}/({L\log(kL/\epsilon)}))$, where $w\sim\left(kL^2/(q\log(kL/\epsilon))\right)^{\frac{1}{3}}$, and $p\doteq \min_{1\leq i\leq L}p_i$.\end{theorem}

\begin{proof}\renewcommand{\IEEEQED}{} The proof follows the same line as in that of Theorem~\ref{thm:DenseCodesRegularBernoulliActualNon-IdenticalGeneral} by implementing the following modifications. Let us replace $p$ and $\epsilon$ with $p/q$ and $\epsilon/q$, respectively. Then, $\varphi=pN_T/wq$, and $r=(1-\gamma^{*})\varphi$, where $\gamma^{*}\sim\sqrt{(1/\dot{\varphi})\ln(w_T q/\dot{\epsilon})}$. Fix a chunk $\omega$. For all $1\leq i\leq L$, and $1\leq j\leq w-L+1$, let $\mathcal{D}(Q_i^j)$, $\mathcal{D}_p(Q_i^j)$, and $r_{ij}$ be defined as before, but only restricted to the packets pertaining to the chunk $\omega$. Similarly as before, for all $i,j$, $\mathcal{D}(Q_i^j)$ can be lower bounded as follows: for all $1\leq j\leq w-L+1$, $\mathcal{D}(Q_1^j)\geq rj$, and for all other values of $i,j$, $\mathcal{D}(Q_i^j)$ fails to be larger than $rj-j(1+o(1))\log(w_T q/\epsilon)$, w.p. b.a.b. $ij\dot{\epsilon}/w_T q$, so long as \begin{equation}\label{eq:Temp15} wq\log\frac{w_T q}{\epsilon}=o(p N_T).\end{equation} Thus the number of dense packets pertaining to the chunk $\omega$ at the sink node fails to be larger than \begin{eqnarray}\label{eq:Temp16}
 \lefteqn{\frac{p N_T}{q} - O\left(\frac{p N_T L}{wq}\right) -  } \nonumber\\
   && O\left(\sqrt{\frac{p N_T w}{q} \log\frac{w q L}{\epsilon}}\right)-O\left(w\log\frac{w q L}{\epsilon}\right),
\end{eqnarray} w.p. b.a.b. $\epsilon/q$. In order to maximize~\eqref{eq:Temp16} subject to condition~\eqref{eq:Temp15}, we specify $w$ by \[\sqrt[3]{\frac{p N_T L^2}{q\log(p N_T L/\epsilon)}}.\] Now let us assume that $N_T$ is $(1+o(1))k/p$. By replacing $\epsilon$ with $\dot{\epsilon}$, in the preceding results, and by replacing $k$ and $\epsilon$ with $k/q$ and $\dot{\epsilon}/q$, respectively, in Lemma~\ref{lem:DenseRankProb}, it follows that the sink node fails to decode the chunk $\omega$ w.p. b.a.b. $\epsilon/q$, so long as $N_T$ is larger than \begin{dmath}\label{eq:Temp17} \frac{1}{p}\left(k+(1+o(1))\left(\frac{k L}{w}+\sqrt{k\left(wq\log\frac{w q L}{\epsilon}\right)}+wq\log\frac{w q L}{\epsilon}\right)\right).\end{dmath} Taking a union bound over all the chunks, it follows that the sink node fails to decode all the chunks w.p. b.a.b. $\epsilon$, so long as $N_T$ is larger than~\eqref{eq:Temp17}. To ensure that the lower bound on $N_T$ is $(1+o(1))k/p$, all the terms in~\eqref{eq:Temp17}, excluding the first one, need to be $o(k/p)$. This condition is met so long as $q$ is \[\hspace{2.65 in}o\left(\frac{k}{L\log(kL/\epsilon)}\right).\hspace{2.65 in}\IEEEQEDopen\]\end{proof}

\begin{theorem}\label{thm:CapAchAveCodDelGeneral} The average coding delay of a CC with $q$ chunks over a network similar to Theorem~\ref{thm:CapAchCodDelGeneral} is larger than \begin{dmath*}\frac{1}{p}\left(k+(1+o(1))\left(\frac{k L}{w}+wq\log\frac{w q L}{\epsilon}\right)\right)\end{dmath*} w.p. b.a.b. $\epsilon$, so long as $q=o({k}/({L\log(kL/\epsilon)}))$, where $w\sim\left(kL/(q\log(kL/\epsilon))\right)^{\frac{1}{2}}$.\end{theorem}

\begin{proof}The proof is similar to that of Theorem~\ref{thm:CapAchCodDelGeneral}, except that $r$ needs to be replaced with $\varphi$. This implies that the third term in~\eqref{eq:Temp16} disappears. Thus by specifying $w$ with \[\sqrt{\frac{p N_T L}{q\log(p N_T L/\epsilon)}}\] in order to maximize~\eqref{eq:Temp16}, excluding the third term, subject to condition~\eqref{eq:Temp15}, it follows that the sink node fails to decode all the chunks w.p. b.a.b. $\epsilon$, so long as $N_T$ is larger than \begin{dmath}\label{eq:Temp18} \frac{1}{p}\left(k+(1+o(1))\left(\frac{k L}{w}+wq\log\frac{w q L}{\epsilon}\right)\right).\end{dmath} The rest of the proof follows that of Theorem~\ref{thm:CapAchCodDelGeneral}.\end{proof}

In the case of unequal success parameters, the coding delay and the average coding delay are upper bounded as follows.

\begin{theorem}\label{thm:CapAchCodDelSpecial} The coding delay of a CC with $q$ chunks over a line network of $L$ links with deterministic regular traffics and Bernoulli losses with unequal parameters $\{p_i\}$ is larger~than \begin{dmath*}\frac{1}{p}\left(k+(1+o(1))\left(\frac{k L}{w}+\sqrt{k\left(wq\log\frac{w q L}{\epsilon}\right)}\right)\right)\end{dmath*} w.p. b.a.b. $\epsilon$, so long as $q=o\left(\gamma^3_{e} {k}/({L\log(kL/\epsilon)})\right)$, where $w\sim\gamma_e\left(kL^2/(q\log(kL/\epsilon))\right)^{\frac{1}{3}}$, $p\doteq \min_{1\leq i\leq L}p_i$, $\gamma_e\doteq\min_{1<i\leq L} \gamma_{e_i}$, and $\gamma_{e_i}\doteq |p_i-p_{i-1}|$.\end{theorem}

\begin{proof}\renewcommand{\IEEEQED}{} Fix a chunk $\omega$. By replacing $p$ and $\epsilon$ with $p/q$ and $\epsilon/q$, respectively, in the proof of Theorem~\ref{thm:DenseCodesRegularBernoulliActualNon-Identical}, it follows that the number of dense packets pertaining to the chunk $\omega$ at the sink node fails to be larger than \begin{eqnarray}\label{eq:Temp19}
 \lefteqn{\frac{p N_T}{q} - O\left(\frac{p N_T L}{wq}\right) -  } \nonumber\\
   && O\left(\sqrt{\frac{p N_T w}{q} \log\frac{w q L}{\epsilon}}\right),
\end{eqnarray} w.p. b.a.b. $\epsilon/q$, so long as \begin{equation}\label{eq:Temp20} wq\log\frac{w_T q}{\epsilon}=o(\gamma_e p N_T).\end{equation} The rest of the proof is similar to that of Theorem~\ref{thm:CapAchCodDelGeneral}, except that~\eqref{eq:Temp19} excludes the last term in~\eqref{eq:Temp16}, and the choice of $w$ needs to satisfy condition~\eqref{eq:Temp20}, instead of condition~\eqref{eq:Temp15}. By specifying $w$ with \[\sqrt[3]{\frac{\gamma^3_e p N_T L^2}{q\log(p N_T L/\epsilon)}}\] in order to maximize~\eqref{eq:Temp19} subject to condition~\eqref{eq:Temp20}, it follows that the sink node fails to decode all the chunks w.p. b.a.b. $\epsilon$, so long as $N_T$ is larger than \begin{dmath}\label{eq:Temp21} \frac{1}{p}\left(k+(1+o(1))\left(\frac{k L}{w}+\sqrt{k\left(wq\log\frac{w q L}{\epsilon}\right)}\right)\right).\end{dmath} In~\eqref{eq:Temp21}, each term, except the largest one, needs to be $o(k/p)$, and this condition is met so long as $q$ is \[\hspace{2.65 in}o\left(\frac{\gamma^3_e k}{L\log(kL/\epsilon)}\right).\hspace{2.65 in}\IEEEQEDopen\]\end{proof}

\begin{theorem}\label{thm:CapAchAveCodDelSpecial} The average coding delay of a CC with $q$ chunks over a network similar to Theorem~\ref{thm:CapAchCodDelSpecial} is larger than \begin{dmath*}\frac{1}{p}\left(k+(1+o(1))\left(\frac{k L}{w}\right)\right)\end{dmath*} w.p. b.a.b. $\epsilon$, so long as $q=o\left(\gamma_e {k}/(f(k){L\log(kL/\epsilon)})\right)$, where $w\sim\gamma_e k/(q f(k)\log(kL/\epsilon))$, and $f(k)$ goes to infinity, as $k$ goes to infinity, such that $f(k)=o(\gamma_e k/(\log(kL/\epsilon)))$.\end{theorem}

\begin{proof}\renewcommand{\IEEEQED}{} The proof follows the same line as that of Theorem~\ref{thm:CapAchCodDelGeneral}, except that the choice of $w$ needs to maximize \begin{equation}\label{eq:Temp22} \frac{p N_T}{q} - O\left(\frac{p N_T L}{wq}\right)\end{equation} subject to condition~\eqref{eq:Temp20}. To do so, we specify $w$ by \[\frac{\gamma_e p N_T}{q f(p N_T)\log(p N_T L/\epsilon)},\] where $f(n)$ goes to infinity, as $n$ goes to infinity, such that $f(n)=o(\gamma_e n/(\log(n L/\epsilon)))$. The sink node fails to decode all the chunks w.p. b.a.b. $\epsilon$, so long as $N_T$ is larger than \begin{dmath}\label{eq:Temp23} \frac{1}{p}\left(k+(1+o(1))\left(\frac{k L}{w}\right)\right).\end{dmath} The second term in~\eqref{eq:Temp23} needs to be $o(k/p)$, and this condition is met so long as $q$ is \[\hspace{2.5in}o\left(\frac{\gamma_e k}{f(k)L\log(kL/\epsilon)}\right).\hspace{2.5in}\IEEEQEDopen\]\end{proof}

\vspace{-.35 cm}
\subsection{CC with Precoding: Capacity-Approaching with A Gap}\label{subsec:CCCAPP}
By the results of Section~\ref{subsec:CCCACH}, one can conclude that CC are not capacity-achieving if the size of the chunks does not comply with condition $\alpha=\omega({L\log(kL/\epsilon)})$.\footnote{For non-negative functions $f(n)$ and $g(n)$, we write $f(n)=\omega(g(n))$, if and only if $\lim_{n\rightarrow\infty}f(n)/g(n)=\infty$.} The analysis of Section~\ref{subsec:DC} further does not apply to CC with chunks of small sizes violating the above condition. From a computational complexity perspective, CC with chunks of smaller sizes are, however, of more practical interest (e.g., linear-time CC with constant-size chunks). In the following,~we study CC with chunks of a size constant in the message~size.

Let $\{p_i\}_{1\leq i\leq L}$ be an arbitrary sequence of success parameters, and let $p\doteq \min_{1\leq i\leq L}p_i$. Let the size of the chunks $\alpha$ ($=k/q$) be a constant in the message size $k$, i.e., $\alpha=O(1)$. Fix a chunk, and focus on the packets pertaining to that chunk. Let the time interval $(0,N_T]$ and its $w$ disjoint partitions be defined as before in~Section~\ref{subsec:DC}. Let $\varphi_{ij}$ be the number of packets (pertaining to the given chunk) in the partition $I_{ij}$, and $\varphi_i$ be the expected value of $\varphi_{ij}$. Let $\varphi\doteq \min_{1\leq i\leq L}\varphi_i$. Then, $\varphi_i=p_i N_T/wq$, and $\varphi=pN_T/wq$. Let $N_T=(1+\gamma_c)k/p$, where $0<\gamma_c<1$ is an arbitrarily small constant. By replacing $N_T$ with $(1+\gamma_c)k/p$, $\varphi=(1+\gamma_c)\alpha/w$, and $\varphi=O(1)$, as $w$ is a constant (otherwise, $\varphi$ goes to $0$, as $N_T$ goes to infinity).

By applying the Chernoff bound, it can be shown that $\Pr\{\varphi_{ij}<(1-\gamma^{*})\varphi\}\leq e^{-{\gamma^{*}}^2\dot{\varphi}}$, for every $0<\gamma^{*}<1$. Taking $e^{-{\gamma^{*}}^2\dot{\varphi}}\leq \dot{\gamma_b}/w_T$, it follows that $\varphi_{ij}$ is not larger than or equal to $r\doteq(1-\gamma^{*})\varphi$ w.p. b.a.b. $\dot{\gamma_b}/w_T$, where $\gamma^{*}$ is the smallest real number satisfying $\gamma^{*}\geq \sqrt{(1/\dot{\varphi})\ln(w_T/\dot{\gamma_b})}$, such that $r$ is an integer ($\gamma^{*}=O(1)$). Taking a union bound over all the active partitions of all links, it follows that $\varphi_{ij}$ is not larger than or equal to $r$ w.p. b.a.b. $\dot{\gamma_b}$.

%  (otherwise, if $\gamma_b$ goes to $0$, $\gamma^{*}$ goes to infinity, and violates the condition $0<\gamma^{*}<1$)

Let $\mathcal{D}(Q_{i}^j)$ be the number of dense packets pertaining to the given chunk in the first $j$ active partitions over the $i\textsuperscript{th}$ link.

By applying Lemma~\ref{lem:HorizontalT}, it can be shown that: (i) for all $1\leq j\leq w-L+1$, $\mathcal{D}(Q_1^j)\geq rj$, (ii) for all $1<i\leq L$, the inequality $\mathcal{D}(Q_i^1)\geq r-\log(w_T/\dot{\gamma_b})$ fails w.p. b.a.b. $i\dot{\gamma_b}/w_T$, and (iii) for all the other $i,j$, the inequality $\mathcal{D}(Q_i^j)\geq r-j\log(w_T/\dot{\gamma_b})-\log((j+1)w_T/\dot{\gamma_b})$ fails w.p. b.a.b. $ij\dot{\gamma_b}/w_T$, so long as \begin{equation}\label{eq:Temp10}\alpha = \Omega \left(w^2\log\frac{w_T}{{\gamma_b}}\right).\end{equation}

By using the above results, it follows that the number of dense packets pertaining to the given chunk at the sink node fails to be lower bounded by \begin{dmath}\label{eq:Temp11} \frac{w_T\varphi}{L} - O\left(\frac{w_T}{L}\sqrt{\varphi\log\frac{w_T}{\gamma_b}}\right)-O\left(\frac{w_T}{L}\log\frac{w_T}{\gamma_b}\right)\end{dmath} w.p. b.a.b. $\gamma_b$. The lower bound is non-negative so long as $\alpha = \Omega\left(w\log({w_T}/{\gamma_b})\right)$, and this condition holds so long as condition~\eqref{eq:Temp10} holds. We specify $w$ by $\sqrt[3]{\alpha L^2/\log(\alpha L/\gamma_b)}$ to maximize~\eqref{eq:Temp11}. By replacing $w$ in~\eqref{eq:Temp10}, it can be rewritten as \begin{equation}\label{eq:Temp12} \alpha=\Omega\left(L^4\log\frac{L}{\gamma_b}\right).\end{equation}

By replacing $\gamma_b$ with $\dot{\gamma_b}$, and by applying Lemma~\ref{lem:DenseRankProb}, it follows that the sink node fails to decode the given chunk w.p. b.a.b. $\gamma_b$, so long as~\eqref{eq:Temp11} is larger than $\alpha+\log({1}/{\dot{\gamma_b}})$. By replacing our choice of $w$ in~\eqref{eq:Temp11}, it can be seen that, excluding the first term, the second term dominates the rest. By replacing $\varphi$ with $(1+\gamma_c)\alpha/w$, and by using the properties of the notation $\Omega(.)$, the decoding condition becomes \begin{equation}\label{eq:Temp13}\alpha=\Omega\left(\frac{L}{\gamma^3_c}\log\frac{L}{\gamma_b\gamma_c}\right).\end{equation} Thus, the given chunk is undecodable w.p. b.a.b. $\gamma_b$, so long as both conditions~\eqref{eq:Temp12} and~\eqref{eq:Temp13} are met. In other words, the expected fraction of undecodable chunks is bounded from above by $\gamma_b$. By using a martingale argument similar to the one in~\cite{HBJ:2011}, the concentration of the fraction of undecodable chunks around the expectation can be shown as follows.

%The proof is omitted to avoid repetition.

% (by constructing a martingale sequence over the number of undecodable chunks)

\begin{lemma}\label{lem:Concentration}By applying a CC with chunks of size $\alpha$, satisfying both conditions~\eqref{eq:Temp12} and~\eqref{eq:Temp13}, the fraction of undecodable chunks at the sink node until time $N_T = (1+\gamma_c)k/p$ is larger than $(1+\gamma_a)\gamma_b$, w.p. b.a.b. $\epsilon$, so long as \begin{equation}\label{eq:Temp14}{\alpha^2}/{\gamma^2_a\gamma^2_b}=o({k}/{\log({1}/{\epsilon})}),\end{equation} where $0<\gamma_a,\gamma_b,\gamma_c<1$ are arbitrary constants.\end{lemma}

By the result of Lemma~\ref{lem:Concentration}, the fraction of chunks which are not decodable until time $N_T$ becomes larger than $(1+\gamma_a)\gamma_b$, w.p. b.a.b. $\epsilon$. Since $\gamma_a,\gamma_b$ are non-zero constants, a CC, alone, does not decode all the chunks. However, the completion of decoding of all the chunks is guaranteed by devising a proper precoding scheme~\cite{HBJ:2011}. The precoding works as follows: The set of $k$ message vectors at the source node constitute the input of a capacity-achieving (c.-a.) erasure code, called \emph{precode}. The rate of the precode is $1-(1+\gamma_a)\gamma_b$ (i.e., the precode decoder can correct up to a fraction $(1+\gamma_a)\gamma_b$ of erasures), and the number of the coded packets at the output of the precode, called \emph{intermediate packets}, is $\left(1+(1+\gamma_a)\gamma_b+O(\gamma^2_b)\right)k$. By applying a CC with chunks of size $\alpha$, satisfying conditions~\eqref{eq:Temp12},~\eqref{eq:Temp13} and~\eqref{eq:Temp14}, the fraction of the intermediate packets that are not recoverable at the output of the CC decoder until time $(1+\gamma_c)\left(1+(1+\gamma_a)\gamma_b+O(\gamma^2_b)\right)\frac{k}{p}$ is larger than $(1+\gamma_a)\gamma_b$, w.p. b.a.b. $\epsilon$. Then, the precode decoder can recover all the $k$ message vectors from the set of recovered intermediate packets. Therefore, the coding delay of a CC with precoding (CCP) is upper bounded as follows.

\begin{theorem}\label{thm:CapAppCodDelGeneral} The coding delay of a CCP with chunks of size $\alpha$ and a c.-a. erasure code of rate $1-\gamma_a$, over a line network of $L$ links with deterministic regular traffics and Bernoulli losses with parameters $\{p_i\}$ is larger than $(1+\gamma_c)\left(1+(1+\gamma_a)\gamma_b+O(\gamma^2_b)\right)\frac{k}{p}$, w.p. b.a.b.~$\epsilon$, so long~as \[\alpha=\Omega\left(\left\{\left(\frac{L}{\gamma^3_{c}}\log\frac{L}{\gamma_b\gamma_c}\right),\left(L^4 \log \frac{L}{\gamma_b}\right)\right\}\right),\] and $\alpha^2/\gamma^2_a\gamma^2_b=o(k/\log(1/\epsilon))$, where $0<\gamma_a,\gamma_b,\gamma_c<1$ are arbitrary constants, and $p\doteq \min_{1\leq i\leq L}p_i$.\end{theorem}

In the case of the average coding delay of a CC with precoding, the following can be shown similar to Theorem~\ref{thm:CapAppCodDelGeneral} by replacing $r$ with $\varphi$, and hence the proof is omitted.

\begin{theorem}\label{thm:CapAppAveCodDelGeneral} The average coding delay of a CCP with chunks of size $\alpha$ and a c.-a. erasure code of rate $1-\gamma_a$, over a network similar to Theorem~\ref{thm:CapAppCodDelGeneral} is larger than $(1+\gamma_c)\left(1+(1+\gamma_a)\gamma_b+O(\gamma^2_b)\right)\frac{k}{p}$, w.p. b.a.b. $\epsilon$, so long as \[\alpha=\Omega\left(\frac{L}{\gamma_{c}}\log\frac{L}{\gamma_b \gamma_c}\right),\] and $\alpha^2/\gamma^2_a\gamma^2_b=o(k/\log(1/\epsilon))$, where $0<\gamma_a,\gamma_b,\gamma_{c}<1$ are arbitrary constants.\end{theorem}

\begin{table*}
  \vspace{-.05cm}
  \caption{Comparison of Chunked Codes over Line Networks with Various Traffics}
  \vspace{-.25cm}
  \hspace{-0.25in}
    \begin{tabular}{|p{1.275cm}|p{1cm}|c|c|c|c|}
    \hline
    \multirow{2}{*}{\hspace{.25 cm}\vspace{-.1cm}Traffic} & \multirow{2}{*}{\hspace{-.275 cm}$\vspace{-.1cm}\begin{array}{c} \text{Success} \\ \text{Parameters} \end{array}$} & $\begin{array}{c} \text{Overhead } \text{(}\eta\text{)} \\ \text{and} \end{array}$ & \multirow{2}[4]{*}{\vspace{.1cm} $\begin{array}{c} \text{Size of Chunks} \\ \text{(}\alpha\text{)}\end{array}$} & \multirow{2}[4]{*}{\vspace{.1cm}$w$} & \multirow{2}[4]{*}{\vspace{.1cm}Comments} \\
      &   & \hspace{.1cm}Average Overhead ($\bar{\eta}$)&   &   &  \\
    \hline
    $\begin{array}{c}\hspace{-.3cm} \text{Arbitrary} \\ \hspace{-.3cm}  \text{Deterministic}\end{array}$ & \hspace{.35cm} -  & $\eta=\bar{\eta}=O\left(kL\left(\frac{1}{\alpha} \log\frac{kL}{\epsilon}\right)^{\frac{1}{3}}\right)$ & $\omega\left({L^3\log\frac{kL}{\epsilon}}\right)$ &  -  & \multirow{5}[10]{*}{\vspace{-0.75cm}$\begin{array}{c} m= \frac{kw}{\alpha}\log\left(\frac{kLw}{\alpha\epsilon}\right) \\ f(k)=o\left(\frac{\gamma_e k}{\log\frac{kL}{\epsilon}}\right) \\ \lim_{k\rightarrow\infty}f(k)=\infty \\ \gamma_{e_i}=|p_i-p_{i-1}|\\ \gamma_e=\min_{1<i\leq L}\gamma_{e_i}\\ p=\min_{1\leq i\leq L}p_i \end{array}$} \\
\cline{1-5}  \multirow{4}{*}{\vspace{-0.75cm}$\begin{array}{c}\hspace{-.335cm}\text{Deterministic} \\ \hspace{-.335cm}\text{Regular} \\ \hspace{-.335cm}\text{Transmissions} \\ \hspace{-.335cm}\text{and} \\ \hspace{-.335cm}\text{Bernoulli} \\ \hspace{-.335cm}\text{Losses} \end{array}$} & \multirow{2}{*}{\vspace{-.25cm}Arbitrary} & $\eta = \frac{1}{p}\left((1+o(1))\left(\frac{kL}{w}+k^{\frac{1}{2}}m^{\frac{1}{2}}+m\right)\right)$ & \multirow{2}[4]{*}{\vspace{-.15cm}$\omega\left({L\log\frac{kL}{\epsilon}}\right)$} & $\left(\frac{\alpha L^2}{\log\frac{kL}{\epsilon}}\right)^{\frac{1}{3}}$ &  \\
\cline{3-3}\cline{5-5}      &   & $\bar{\eta}=\frac{1}{p}\left((1+o(1))\left(\frac{kL}{w}+m\right)\right)$ &   & $\left(\frac{\alpha L}{\log\frac{kL}{\epsilon}}\right)^{\frac{1}{2}}$ & \\
\cline{2-5}      & \multirow{2}{*}{\vspace{-.25cm}Unequal} & $\eta = \frac{1}{p}\left((1+o(1))\left(\frac{kL}{w}+k^{\frac{1}{2}}m^{\frac{1}{2}}\right)\right)$ & $\omega\left(\frac{L}{\gamma^3_e} \log\frac{kL}{\epsilon}\right)$ & $\left(\frac{\gamma^3_e\alpha L^2}{\log\frac{kL}{\epsilon}}\right)^{\frac{1}{3}}$ & \\
\cline{3-5}      &   & $\bar{\eta}=\frac{1}{p}\left((1+o(1))\left(\frac{kL}{w}\right)\right)$ & $\omega\left(f(k)\left(\frac{L}{\gamma_e} \log\frac{kL}{\epsilon}\right)\right)$ & $\frac{1}{f(k)}\left(\frac{\gamma_e\alpha}{\log\frac{kL}{\epsilon}}\right)$ & \\
    \hline
    \end{tabular}
  \label{tab:TableI}
\end{table*}

\begin{table*}
\vspace{-.1cm}
\caption{Comparison of Chunked Codes with Precoding (A Capacity-Achieving Erasure Code) over Line Networks with Various Traffics}
\vspace{-.25cm}
\hspace{-0.275in}
    \begin{tabular}{|p{1.275cm}|p{1cm}|c|c|c|c|}
    \hline
    \multirow{2}{*}{\hspace{.25 cm}\vspace{-.1cm}Traffic} & \multirow{2}{*}{\hspace{-.275 cm}$\vspace{-.1cm}\begin{array}{c} \text{Success} \\ \text{Parameters} \end{array}$} & $\begin{array}{c} \text{Overhead } \text{(}\eta\text{)} \\ \text{and} \end{array}$ & \multicolumn{2}{c|}{\multirow{2}{*}{\vspace{-.1cm}$\begin{array}{c} \text{Size of Chunks} \\ \text{(}\alpha\text{)}\end{array}$}} & \multirow{2}{*}{\vspace{-.1cm}Comments} \\      &   & \hspace{.1cm}Average Overhead ($\bar{\eta}$) & \multicolumn{2}{c|}{} &  \\
    \hline
    $\begin{array}{c}\hspace{-.3cm} \text{Arbitrary} \\ \hspace{-.3cm}  \text{Deterministic}\end{array}$ & \hspace{.35cm} - & \hspace{.1cm}$\eta=\bar{\eta}=\gamma_o k$ & $\Omega\left(\frac{L^3}{\gamma^3_c}\log\frac{L}{\gamma_b\gamma_c}\right)$ & \multicolumn{1}{r|}{\multirow{5}[10]{*}{\vspace{-.25cm}$o\left(\sqrt{\frac{\gamma^2_a\gamma^2_b k}{\log\frac{1}{\epsilon}}}\right)$}} & \multirow{5}[10]{*}{\vspace{-.1cm}$\begin{array}{c} 0<\gamma_a,\gamma_b,\gamma_c<1 \\ \{\gamma_a,\gamma_b,\gamma_c\}=O(1) \\ \gamma_o=\gamma_c+(1+\gamma_c)\gamma'_o \\ \gamma'_o = (1+\gamma_a)\gamma_b+O(\gamma^2_b) \\  \gamma_{e_i}=|p_i-p_{i-1}| \\ \gamma_e=\min_{1<i\leq L}\gamma_{e_i} \\ p=\min_{1\leq i\leq L}p_i \end{array}$} \\
\cline{1-2}\cline{3-3}\cline{4-4}    \multirow{4}{*}{$\begin{array}{c}\hspace{-.335cm}\text{Deterministic} \\ \hspace{-.335cm}\text{Regular} \\ \hspace{-.335cm}\text{Transmissions} \\ \hspace{-.335cm}\text{and} \\ \hspace{-.335cm}\text{Bernoulli} \\ \hspace{-.335cm}\text{Losses} \end{array}$} & \multirow{2}[4]{*}{Arbitrary} &  {$\eta=\gamma_o\frac{k}{p}$} & $\Omega\left(\left\{\left(\frac{L}{\gamma^3_c}\log\frac{L}{\gamma_b\gamma_c}\right),\left(L^4\log\frac{L}{\gamma_b}\right)\right\}\right)^{\textcolor[rgb]{1.00,1.00,1.00}{\frac{1}{2}}}_{\textcolor[rgb]{1.00,1.00,1.00}{\frac{1}{2}}}$ &   &  \\ \cline{3-3}
\cline{4-4}      &   &  $\bar{\eta}=\gamma_o\frac{k}{p}$ & $\Omega\left(\frac{L}{\gamma_c}\log\frac{L}{\gamma_b\gamma_c}\right)^{\textcolor[rgb]{1.00,1.00,1.00}{\frac{1}{2}}}_{\textcolor[rgb]{1.00,1.00,1.00}{\frac{1}{2}}}$ &   &  \\ \cline{3-3}
\cline{2-2}\cline{4-4}      & \multirow{2}[4]{*}{Unequal} &  $\eta=\gamma_o\frac{k}{p}$ & $\Omega\left(\left\{\left(\frac{L}{\gamma^3_c}\log\frac{L}{\gamma_b\gamma_c}\right),\left(\frac{L}{\gamma^3_e}\log\frac{L}{\gamma_b\gamma_e}\right)\right\}\right)^{\textcolor[rgb]{1.00,1.00,1.00}{\frac{1}{2}}}_{\textcolor[rgb]{1.00,1.00,1.00}{\frac{1}{2}}}$ &   &  \\ \cline{3-3}
\cline{4-4}      &   & $\bar{\eta}=\gamma_o\frac{k}{p}$  & $\Omega\left(\frac{L}{\gamma^2_e\gamma_c}\log\frac{L}{\gamma_b\gamma_c}\right)^{\textcolor[rgb]{1.00,1.00,1.00}{\frac{1}{2}}}_{\textcolor[rgb]{1.00,1.00,1.00}{\frac{1}{2}}}$ &   &  \\ \cline{3-3}
    \hline
    \end{tabular}
  \label{tab:TableII}
  \vspace{-.45cm}
\end{table*}

In the special case of unequal success parameters, the coding delay and the average coding delay of CC with precoding are upper bounded as follows. The proofs follow the same line as in the general case except that a new set of conditions needs to be satisfied based on the assumption that no two success parameters are equal.

\begin{theorem}\label{thm:CapAppCodDelSpecial} The coding delay of a CCP with chunks of size $\alpha$ and a c.-a. erasure code of rate $1-\gamma_a$, over a line network of $L$ links with deterministic regular traffics and Bernoulli losses with unequal parameters $\{p_i\}$ is larger than $(1+\gamma_c)\left(1+(1+\gamma_a)\gamma_b+O(\gamma^2_b)\right)\frac{k}{p}$, w.p. b.a.b. $\epsilon$, so long as \[\alpha=\Omega\left(\left\{\left(\frac{L}{\gamma^3_{c}}\log\frac{L}{\gamma_b\gamma_c}\right),\left(\frac{L}{\gamma^3_{e}} \log \frac{L}{\gamma_e\gamma_b}\right)\right\}\right),\] and $\alpha^2/\gamma^2_a\gamma^2_b=o(k/\log(1/\epsilon))$, where $0<\gamma_a,\gamma_b,\gamma_c<1$ are arbitrary constants, $p\doteq \min_{1\leq i\leq L}p_i$, $\gamma_e\doteq\min_{1<i\leq L} \gamma_{e_i}$, and $\gamma_{e_i}\doteq |p_i-p_{i-1}|$.\end{theorem}

\begin{proof}Let us assume $p_1>p_2>\cdots> p_L$, without loss of generality. Let $p\doteq \min_{1\leq i\leq L} p_i$, $\gamma_e\doteq \min_{1<i\leq L}\gamma_{e_i}$, and $\gamma_{e_i}\doteq |p_i-p_{i-1}|$. Fix a chunk. Let $r_i\doteq (1-\gamma^{*}_i)\varphi_i$, where $\varphi_i=p_iN_T/wq$ and $\gamma^{*}_i\sim\sqrt{(1/\dot{\varphi_i})\log(w_T/\dot{\gamma_b})}$, and $0<\gamma_b<1$ is an arbitrary constant. Let $\varphi_{ij}$ be the number of packets (pertaining to the given chunk) in the partition $I_{ij}$ (the $j\textsuperscript{th}$ partition pertaining to the $i\textsuperscript{th}$ link), where the time interval $(0,N_T]$ is split into $w$ partitions of length $N_T/w$, and let $\varphi_i$ be the expected value of $\varphi_{ij}$. For all $i,j$, suppose that $\varphi_{ij}$ is larger than or equal to $r_i$. Let $N_T=(1+\gamma_c)k/p$, where $0<\gamma_c<1$ is an arbitrarily small constant. By replacing $N_T$ with $(1+\gamma_c)k/p$, $\varphi_i=(1+\gamma_c)p_i\alpha/pw$, and $\varphi=O(1)$, similar to that in the proof of Theorem~\ref{thm:CapAppCodDelGeneral}.

Similarly as before, for all $j$, $\mathcal{D}(Q_1^j)\geq r_1 j$. For any other values of $i,j$, by applying Lemma~\ref{lem:HorizontalT}, it can be shown that the inequality $\mathcal{D}(Q_i^j)\geq r_i j$ fails w.p. b.a.b. $ij\dot{\gamma_b}/w_T$, so long as \begin{equation}\label{eq:Temp24} \alpha=\Omega\left(\frac{w}{\gamma_e^2}\log\frac{w_T}{\gamma_b}\right).\end{equation}

Let $\varphi$, $\gamma^{*}$ and $r$ denote $\varphi_L$, $\gamma^{*}_L$ and $r_L$, respectively. Thus, the number od dense packets pertaining to the given chunk at the sink node fails to be larger than \begin{dmath}\label{eq:Temp25}\alpha(1+\gamma_c)-O\left(\frac{\alpha L}{w}\right)-O\left(\sqrt{\alpha w\log\frac{w_T}{\gamma_b}}\right).\end{dmath} We specify $w$ by \[\left(\frac{\alpha L^2}{\log(w_T/\gamma_b)}\right)^{\frac{1}{3}}\] to maximize~\eqref{eq:Temp25} subject to condition~\eqref{eq:Temp24}. For this choice of $w$, condition~\eqref{eq:Temp24} is met so long as \begin{equation}\label{eq:Temp26}\alpha=\Omega\left(\frac{L}{\gamma^3_e}\log\frac{L}{\gamma_e\gamma_b}\right).\end{equation} By replacing $\gamma_b$ with $\dot{\gamma_b}$ in the preceding results, and substituting $w$ in~\eqref{eq:Temp25}, the result of Lemma~\ref{lem:DenseRankProb} shows that the sink node fails to decode the given chunk w.p. b.a.b. $\gamma_b$, so long as~\eqref{eq:Temp25} is larger than $\alpha+\log(1/\dot{\gamma_b})$. Based on the properties of the notation $\Omega(.)$, the latter condition is met so long as \begin{equation}\label{eq:Temp27}\alpha=\Omega\left(\frac{L}{\gamma^3_c}\log\frac{L}{\gamma_b\gamma_c}\right).\end{equation}
The rest of the proof is similar to the proof of Theorem~\ref{thm:CapAppCodDelGeneral}, except that in this case conditions~\eqref{eq:Temp26} and~\eqref{eq:Temp27} need to be met, instead of conditions~\eqref{eq:Temp12} and~\eqref{eq:Temp13}.\end{proof}

\begin{theorem}\label{thm:CapAppAveCodDelSpecial} The average coding delay of a CCP with chunks of size $\alpha$ and a c.-a. erasure code of rate $1-\gamma_a$, over a network similar to Theorem~\ref{thm:CapAppCodDelSpecial} is larger than $(1+\gamma_c)\left(1+(1+\gamma_a)\gamma_b+O(\gamma^2_b)\right)\frac{k}{p}$, w.p. b.a.b. $\epsilon$, so long as \[\alpha=\Omega\left(\frac{L}{\gamma^2_e\gamma_{c}}\log\frac{L}{\gamma_b \gamma_c}\right),\] and $\alpha^2/\gamma^2_a\gamma^2_b=o(k/\log(1/\epsilon))$, where $0<\gamma_a,\gamma_b,\gamma_{c}<1$ are arbitrary constants.\end{theorem}

\begin{proof}\renewcommand{\IEEEQED}{}The proof follows the same line as that of Theorem~\ref{thm:CapAppCodDelSpecial}, except that the choice of $w$ needs to maximize \begin{equation}\label{eq:Temp28} \alpha(1+\gamma_c)-O\left(\frac{\alpha L}{w}\right)\end{equation} subject to condition~\eqref{eq:Temp24}. To do so, the choice of $w$ needs to be $\Omega(L/\gamma_c)$, and hence, condition~\eqref{eq:Temp24} becomes \[\hspace{2.45in}\alpha=\Omega\left(\frac{L}{\gamma^2_e\gamma_c}\log\frac{L}{\gamma_b\gamma_c}\right).\hspace{2.45in}\IEEEQEDopen\]
\end{proof}

\vspace{-.25 cm}
\section{Poisson Transmissions and Bernoulli Losses}\label{sec:PoissonTrafficCC}
In the case of Bernoulli losses and Poisson transmissions with parameters $\{p_i\}_{1\leq i\leq L}$ and $\{\lambda_i\}_{1\leq i\leq L}$, the points in time at which the arrivals/departures occur over the $i\textsuperscript{th}$ link follow a Poisson process with parameter $\lambda_i p_i$. Thus the number of packets pertaining to a given chunk, in each partition pertaining to the $i\textsuperscript{th}$ link, has a Poisson distribution with the expected value $\lambda_i p_i N_T/wq$. Since the result of Chernoff bound also holds for Poisson random variables, the main results in Section~\ref{sec:BernoulliLossRegularTrafficCC} apply to this case by replacing $p$ with $\lambda p$, where $\{\lambda,p\}\doteq \{\lambda_{\mu},p_{\mu}\}$, and $\mu\doteq\arg\min_{1\leq i\leq L}\lambda_i p_i$.

\vspace{-.25cm}
\section{Discussion}\label{sec:Discussion}
Table~\ref{tab:TableI} shows the upper bounds\footnote{With a slight abuse of language, we refer to the ``upper bound'' on the overhead or the average overhead as the ``overhead'' or the ``average overhead.''} (w.p. of failure b.a.b. $\epsilon$) on the overhead and the average overhead of CC over various traffics for different ranges of the size of the chunks based on the results in Section~\ref{sec:BernoulliLossRegularTrafficCC} and those in~\cite{HBJ:2011}.\footnote{The results of Section~\ref{subsec:CCCACH} and those of Section~\ref{subsec:CCCAPP} were stated in terms of $q$ and $\alpha$, respectively. In this section, for the ease of comparison, the former results are also restated in terms of $\alpha$ by replacing $q$~with~$k/\alpha$.} The traffics are: arbitrary deterministic traffics, or traffics with deterministic regular transmissions and Bernoulli losses. We refer to the latter traffics as the \emph{probabilistic traffics} for simplifying the terminology. The probabilistic traffics are categorized into two sub-categories: traffics with arbitrary success parameters and traffics with unequal success parameters. In the case of arbitrary deterministic traffics, the capacity is $1$, and in the case of probabilistic traffics with success parameters $\{p_i\}_{1\leq i\leq L}$, the capacity is $p$, where $p=\min_{1\leq i\leq L}p_i$. We say that a code is ``capacity-achieving'' (c.-a.) if the ratio of the overhead to $k/p$ goes to $0$, as $k$ goes to infinity. Similarly, a code is ``capacity-achieving on average'' (c.-a.a.) if the ratio of the average overhead to $k/p$ goes to $0$, as $k$ goes to infinity. In Table~\ref{tab:TableI}, the upper (or lower) row in front of each case of success parameters corresponds to a c.-a. (or a c.-a.a.) scenario.

% \footnote{In the case of arbitrary deterministic traffics, the overhead and the average overhead are equal, and are a random variable due to the randomness of the code (there is no randomness in the traffic). In the case of the probabilistic traffics, the (average) overhead is a random variable due to the randomness in (the traffic, but not the code) the code and the traffic.}

% In each case, the number of chunks is upper bounded to ensure that CC are capacity-achieving in the corresponding sense.

In the table, one can see that, for each traffic, the size of the chunks ($\alpha$) has to be sufficiently large so that CC are c.-a. or c.-a.a.. For arbitrary deterministic traffics, the lower bound on $\alpha$ is super-logarithmic in $k$, i.e., $\omega(\log k)$, and super-log-cubic in $L$, i.e., $\omega({L^3\log L})$. For the probabilistic traffics with arbitrary or unequal success parameters, the lower bound on $\alpha$ has a similar growth rate with $k$, but a smaller (super-log-linear) growth rate with $L$, i.e., $\omega({L\log L})$. The coding cost of CC (i.e., the ratio of the number of the coding (packet) operations to $k$), is, on the other hand, linear in $\alpha$. Thus, CC can perform as fast over both the arbitrary deterministic traffics and the probabilistic traffics, but with a lower coding cost (smaller chunks) in the latter case compared to the former.

Moreover, as it can be seen in Table~\ref{tab:TableI}, for both arbitrary deterministic and probabilistic traffics (in each case of arbitrary or unequal success parameters), the overhead grows sub-log-linearly with $k$, i.e., $O(k\log^{\frac{1}{3}} k)$, and decays sub-linearly with $\alpha$, i.e., $O(1/\alpha^{\frac{1}{3}})$. However, for arbitrary deterministic traffics, the overhead grows with $O(L\log^{\frac{1}{3}}L)$, and for the probabilistic traffics, it only grows with $O(L^{\frac{1}{3}}\log^{\frac{1}{3}} L)$. This implies a faster speed of convergence to the capacity in the latter case compared to the former. Similar comparison result can also be observed in terms of the average overhead, except that in the case of unequal success parameters, the average overhead decays linearly with $\alpha$, i.e., $O(1/\alpha)$, but grows poly-log-linearly with $k$, i.e., $O(k\log^2 k)$, for the choice of $f(k)=O(\gamma_e \log k)$, and log-linearly with $L$, i.e., $O(L\log L)$.

Table~\ref{tab:TableII} also shows the results for CC with precoding (CCP) in the scenarios similar to those considered in Table~\ref{tab:TableI}, where the precode is an (capacity-achieving) erasure code of dimension $k$ and rate $1-\gamma_a$. In particular, one can see that CCP are ``capacity-approaching'' or ``capacity-approaching on average'' with an arbitrary small ``non-zero constant'' gap $\gamma_o$ (i.e., the ratio of the overhead or the average overhead to $k/p$ goes to $\gamma_o$, as $k$ goes to infinity) if $\alpha$ is sufficiently large. For simplifying the terminology, we drop the term ``with a non-zero constant gap.'' The upper row (or the lower row) in front of each case of success parameters corresponds to a capacity-approaching (or a capacity-approaching on average) scenario. For arbitrary deterministic traffics, the lower bound on $\alpha$ is constant in $k$, and log-cubic in $L$, i.e., $O(L^3\log L)$. For the probabilistic traffics with arbitrary or unequal success parameters, the lower bound on $\alpha$ is also constant in $k$, but has a smaller (log-linear) growth rate with $L$, i.e., $O(L\log L)$. Thus, in the case of CCP, one can make a conclusion similar to the one made in the case of stand-alone CC, with respect to the arbitrary deterministic and the probabilistic traffics.

% However, one should note that, for all $q$ sufficiently small (i.e., CC are capacity-achieving in each sense, in each case of success parameters), the average overhead in the case of unequal success parameters is still smaller than the one in the case of arbitrary success parameters.

\bibliographystyle{IEEEtran}
\bibliography{RefsII}

\end{document}